\documentclass[]{article}
\usepackage{amsthm}
\usepackage{thmtools}
\usepackage{amsmath}
\usepackage{amssymb}
\usepackage{cases}
\usepackage{enumitem} 
\newtheorem{thm}{Theorem}[section]

\newtheorem*{openq}{Open question}
\theoremstyle{definition}
\declaretheorem[
  style=definition, 
  qed=$\blacksquare$, 
  numberwithin=section
]{definition}
\newtheorem{claim}{Claim}
\newtheorem{lemma}{Lemma}
\newtheorem{proposition}{Proposition}

\title{How far is Kolmogorov-Loveland randomness from Martin-L\"of randomness?}
\author{Tomislav Petrovi\'c}

\begin{document}

\maketitle

\begin{abstract}
We'll call a map from basic to clopen sets of binary sequences an \textit{encoding} if it preserves uniform Lebesgue measure and the disjointedness relation.
For some notion of randomness $R$ we can ask if there is a computable encoding $E$
such that for every non-Martin-L\"of-random (non-MLR) sequence  $\alpha$, the sequences in $E(\alpha)$ are non-$R$. 
Thus we can use $E$ and randomness test for $R$ to construct a randomness test for MLR.
We show that there is no such $E$ for partial computable randomness (PCR), but there is one for Kolmogorov-Loveland randomness.
Indeed, there is a total computable $E$ and a pair of total computable non-monotonic betting strategies such that on every sequence in $E(\alpha)$ at least one of the strategies wins whenever $\alpha$ is non-MLR.
\end{abstract}

\section{Introduction}
The history behind the open question of whether Kolmogorov-Loveland randomness (KLR) is equal to MLR is very long,  
see \cite[Chapter 9]{ShenBook}, for exposition.
MLR was introduced by P. Martin-L\"of in \cite{MartinLof} in terms of effective nullsets.
In \cite{Schnorr}, Schnorr defines randomness in terms of martingales,
making successive bets on the value of the next bit of a binary sequence.
He shows that MLR can be characterized
by a lower semicomputable martingale,
which he finds objectionable,
and introduces computable randomness (CR) that uses totally computable martingales,
For computable martingales, making a bet on a bit immediately determines the capital for both outcomes, 
as opposed to capital being forever approached from below 
without knowing how well the current value
approximates the value of capital in the limit.
If we allow martingales to be partially computable, we obtain PCR, which strictly implies CR, \cite{Klaus}.
In \cite{Muchnik} 
KLR was defined
by lifting the restriction that the bet must be made on the next bit, and instead allowing the position of the bit 
to be chosen.
In \cite[Theorem 9.1]{Muchnik} it is shown that sequences with (ML)randomness deficiency 
lower bounded by an unbounded computable function are non-KLR,
and by \cite[Theorem 9.5]{Muchnik}, KLR strictly implies implies PCR.
More generally,  \cite[Theorem 12]{Merkle} shows that if a sequence can be computably split into
two non-MLR  sequences then it is non-KLR.
In theorem \ref{KLR=_EMLR} we'll show that we can computably encode a sequence, and use KLR test on the encoding,
in order to test whether the sequence is MLR.
This gives us a characterization of MLR that does not use lower semicomputability.
In theorem \ref{thm_no_E_for_PCR}
we'll show that not every notion of randomness has this property, for PCR such an encoding cannot exist.
Whether KLR is in this respect unique is an open question.
\newcommand{\pos}{\mathtt{pos}}
\section{Notation and definitions}
The notation is mostly standard.
A union of disjoint sets is denoted with symbol $\sqcup$.
The negation of bit $v$ is denoted with $\bar{v}$.
A finite sequence of bits is called a (binary) \textit{string}.
The empty string is denoted with $\epsilon$.
We'll call an infinite sequence of bits a (binary) \textit{sequence}.
The $n$-th bit of a sequence $\alpha$ is denoted $\alpha[n]$.
If a string $s$ is a prefix of a string or sequence $s'$ we write $s\preceq s$, we'll also say that $s'$ extends $s$. 
A prefix of a sequence $\alpha$ of length $n$ is denoted $\alpha[1:n]$, a segment of the sequence
between position $n$ and $m$ is denoted $\alpha[n:m]$.
A concatenation of strings $a$ and $s$ is written $as$.
The length of a string $a$ is denoted $|a|$. 
A basic (open) set of sequences extending a string $a$ is denoted with $\tilde{a}$.
A (possibly infinite) union of basic sets is called an open set. 
The complement of an open set is a closed set.
A finite union of basic sets is called clopen (it is both open and closed).
For a set of strings $X$, we'll denote with $\tilde X = \bigcup_{x\in X} \tilde x$.
The subset of a set of strings $X$ that extends a string $x$ we denote with $X|x$.
A finite assignment is a finite sequence of pairs of natural numbers and bits. 
If the pair $(p,v)$ is in a finite assignment $a$, we say that $a$ assigns value $v$ to position $p$.
The length of a finite assignment $a$ is the number of pairs in the sequence and is denoted with $|a|$.
The set of positions to which a finite assignment assigns a value is denoted with $\pos(a)$.
We'll say that a binary sequence $\alpha$ is consistent with a finite assignment $a=(p_1,v_1),\dots,(p_n,v_n)$
if for all $i\leq n$, $\alpha[p_i] = v_i$.
Two finite assignments are consistent if for every position to which they both assign a value, they assign the same value.
Two finite assignments are equivalent if they assign the same values to the same positions (regardless of the order of the pairs in the assignments). For equivalent $f,a$ we write $f\equiv a$.
A finite assignment $f$ and a string $s$ are equivalent if $\pos(f)=[1,|s|]$, and $f$ assigns $s[p]$ to every position $p\in [1,|s|]$.
A finite assignment $f$ refines a finite assignment $a$ if  $f$ assigns to $\pos(a)$ the same values as $a$,
and possibly assigns some values to positions outside  $\pos(a)$.
The (clopen) set of all sequences consistent with a finite assignment $a$ is denoted with $\tilde{a}$.
For consistent $f,a$, $\tilde f$ intersects $\tilde a$. For equivalent $f,a$, $\tilde f = \tilde a$, and
$\tilde f \subseteq \tilde a$ when $f$ refines $a$.
Uniform Lebesgue measure is denoted with $\lambda$, for a string or a finite assignment $a$,
we have $\lambda(\tilde{a})=2^{-|a|}$.
In this paper we will only consider randomness with respect to the uniform Lebesgue measure.

\begin{definition}
We'll say that a function $\phi$ from some domain $X$ to reals is lower semicomputable iff there is a computable
function $f$ from $X,\mathbb{N}$ to rationals, such that $f(x,n+1)\geq f(x,n)$ and $\phi(x)=\lim_n f(x,n)$.
\end{definition}

\begin{definition}
A (possibly partial) function $\mu$ from strings to non-negative reals is called a (probability) \textit{measure} iff $\mu(\epsilon)=1$ and for any string $a$ and bit $v$,
when $\mu(av)$ is defined, then $\mu(a\bar{v})$ is also defined and $\mu(a)=\mu(a0)+\mu(a1)$.
If  $\mu(a)\geq\mu(a0)+\mu(a1)$, $\mu$ is called a \textit{semimeasure}.
The ratio between  measures $\mu$ and $\lambda$ we denote with $\dot\mu$ and call a \textit{martingale}, .
The interpretation is that of a gambler's betting strategy, 
where, having bet on the first $|a|$ bits of a sequence extending $a$,
 a bet is placed on the value of the next bit by
dividing the \textit{capital} $\dot\mu(a)$ between
two outcomes, $a0,a1$.  
For outcome $av$, the capital put on $a\bar v$ is lost, and the capital put on $av$ is doubled.
The martingale \textit{wins} on sequence $\alpha$ if 
the capital over the sequence of bets is unbounded,  
$\sup_n \dot\mu(\alpha{1:n}) = \infty$. 

If $\mu$ is a semimeasure, $\dot\mu$ is called a \textit{supermartingale}.
We will look at \textit{partial computable martingales} where $\mu$ is partial computable,
and at  \textit{lower semicomputable supermartingales} where $\mu$ is a lower semicomputable semimeasure.
\end{definition}

\begin{definition}
Let $\pi$ be a (possibly partial) function
from strings to natural numbers such that 
if $\pi(a)$ is defined then it is also defined for all prefixes of $a$,
and $\pi(a)\neq\pi(a')$, for any prefix $a'$ of $a$.
Such $\pi$ is called a \textit{scan rule}, and we interpret the output as the position of bit to be scanned next.
For a string $s$ consisting of $n$ bits, we'll call the finite assignment 

$(\pi(\epsilon),s[1]),(\pi(s[1]),s[2]),\dots,(\pi(s[1:n-1]),s[n])$ \textit{$\pi$'s finite assignment for $s$} and denote it with $\mathtt{fa}_\pi(s)$.
\end{definition}

\begin{definition}
A \textit{betting strategy} is a pair of scan rule and measure.
Let $A=(\pi,\mu)$ be a betting strategy.
We can view martingales as betting strategies where the scan rule is in-order, $\pi(a)=|a|+1$.
When the scan rule is not restricted to be in order, to highlight this, we say that the betting strategy is \textit{nonmonotonic}.
We will be interested in betting strategies where $\pi$ is not necessarily in-order, and both $\pi$ and  $\mu$ are total computable, and we'll call these strategies \textit{total computable nonmonotonic betting strategies}.

For strategy $A$, we'll write $\mathtt{fa}_A$ and mean $\mathtt{fa}_\pi$. Instead of writing $\pos(\mathtt{fa}_A(s))$ we write $\pos_A(s)$.

Let $\alpha,\sigma$ be sequences such that for all $n$, $\alpha$ is consistent
with $\mathtt{fa}_A(\sigma[1:n])$.
We'll say that $A$ has \textit{capital} 
$\dot\mu(\sigma[1:n])=\frac{\mu(\sigma[1:n])}{\lambda(\tilde{\mathtt{fa}_A(\sigma[1:n]))}}=2^n\mu(\sigma[1:n])$
when betting on $\alpha$ after $n$ bets, and
 wins on $\alpha$ if $\sup_n \dot\mu(\sigma_{1:n}) = \infty$. 

We will be defining strategies recursively.
If a strategy is defined on a string $x$, but not on $x0,x1$, we'll say that $x$ is a \textit{leaf} of the strategy. 
Defining a bet for a leaf $x$ consists of defining the scan rule for $x$, and measure for the two new
leafs $x0,x1$. 
\end{definition}

\begin{definition}[PCR]
We'll say that a sequence is \textit{partially computably random} if no partial computable martingale wins on the sequence.
\end{definition}

\begin{definition}[KLR]
We'll say that a sequence is \textit{Kolmogorov-Loveland random} if no total computable nonmonotonic betting strategy wins on the sequence.
\end{definition}

\begin{definition}[MLR]
We'll say that a sequence is \textit{Martin-L\"of random} if no lower semicomputable supermartingale wins on the sequence. 
A lower semicomputable supermartingale is universal if it  multiplicatively dominates all other lower semicomputable supermartingales.
We pick some universal lower semicomputable supermartingale  and denote it with $u$.  
\end{definition}
That a universal lower semicomputable supermartingale exists was shown in \cite{Schnorr}.

\begin{definition}
We'll call a map from basic to clopen sets of binary sequences an \textit{encoding} if it preserves uniform Lebesgue measure and the disjointedness relation, that is, for any string $a$,  $\lambda(E(\tilde{a}))=\lambda(\tilde{a})$ and $E(\tilde{a})=E(\tilde{a0})\sqcup E(\tilde{a1})$.
For an encoding $E$ and sequence $\alpha$, instead of writing $\bigcap_n E(\tilde\alpha[1:n])$ we'll  just write $E(\alpha)$.
\end{definition}

\section{No computable encoding for PCR}
The proof is similar to the textbook proofs that MLR strictly implies PCR, \cite[Chapter 7.4.2]{Downey}, \cite[Chapter 7.4]{Nies}. 
We'll have an added level of indirection brought by mapping basic sets with an encoding, and we will, like in \cite[Theorem 9.5]{Muchnik}
be constructing a lower semicomputable semimeasure. Apart from having an encoding, our goal is simpler, as we are not concerned with the exact bound on complexity of initial prefixes of the constructed sequence.

\begin{thm}\label{thm_no_E_for_PCR}
There is no computable encoding $E$ such that for every non-MLR sequence $\alpha$ the sequences in $E(\alpha)$ are non-PCR. 
\end{thm}
\begin{proof}
Suppose there is such an $E$.
We'll construct a lower semicomputable semimeasure $\gamma$ so that $\dot\gamma$ wins on some sequence $\alpha$ such that $E(\alpha)$ contains a PCR sequence.
Let $\mu_1,\mu_2,\dots$ be a computable enumeration of partial computable measures.
Let $\hat \mu$ denote the largest capital achieved on a string so far, $\hat \mu (s) = \sup_{\ell\leq |s|} \dot\mu (s[1:\ell])$.

To construct $\gamma$, we first construct two "helper" functions $F,Q$.
$F$ is a partial computable function whose input is
a finite sequence of strings $x=s_1,\dots,s_n$ and outputs a string pair $(a,b)$
where $a$ will have high capital achieved by $\gamma$,
 and $b$  has bounded capital achieved by the first $n$ partial computable martingales
 and is a subset of$E(\tilde a)$.
$Q$ is a partial computable function from finite sequences of strings to rationals that computes bounds on
achieved capital by partial computable martingales.
The strings in the pair $F(x[1:i+1])$ extend the strings in the pair $F(x[1:i])$ and we'll construct
a pair of sequences $(\alpha,\beta)$ such that $\beta\in E(\alpha)$ and $\gamma$ wins on $\alpha$,
whereas every partial computable martingale achieves only bounded capital on $\beta$.

For empty sequence of strings, $\xi$, we define $F(\xi)=(\epsilon,\epsilon)$ and $Q(\xi)=0$.
Let $x$ be a finite sequence of strings, $x=s_1,\dots,s_n$,
and let $x'$ be $x$ extended by string $s_{n+1}$,  
$x'=s_1,\dots,s_n,s_{n+1}=x,s_{n+1}$.
If $F(x)$ or $Q(x)$ is undefined then $F(x')$ and $Q(x')$ are both undefined (for any $s_{n+1}$).
Suppose both $F(x)$ and $Q(x)$ are defined.
Let $(a,b)=F(x)$.
For $i\leq n$, denote with $q_i=Q(x[1:i])$. 
Denote with $C_n(b)$ the weighted sum of achieved capitals
by the first $n$ partial computable martingales,
 $C_n(b) = \sum_{i\leq n} q_i \hat \mu_i(b)$.
If for $F(x),Q(x)$ the condition
\begin{equation}
C_n(b)< 2\text{ and }\tilde{b}\subseteq E(\tilde a). 
\label{CA}
\end{equation}
is not satisfied, $F(x'),Q(x')$ are undefined (for any $s_{n+1}$).

We'll now define how $F(x')$ and $Q(x')$ are computed, when $F(x),Q(x)$ are defined and condition \eqref{CA} is satisfied. 
Based on the string $s_{n+1}$, we'll make an additional assumption $H(x')$.
\begin{numcases}{H(x') = }
s_{n+1}=\epsilon, &for all  $b'\succeq b$, if $C_n(b') < 2$, $\mu_{n+1}(b')$ is defined\label{allb} \\
s_{n+1}\neq\epsilon, &$\mu_{n+1}$ is not defined on $b'=bs_{n+1}$ and $C_n(b') < 2$\label{existsb}
\end{numcases}

If $s_{n+1}\neq\epsilon$,
we set $q_{n+1} = Q(x') = 0$. 
Let $b' = bs_{n+1}$. By assumption $H(x')$, $C_{n+1}(b')<2$.
We then compute every measure $\mu_i, i\leq n$ for which $q_i>0$ on all continuations of $b'$ until we find
some $b''$ such that there is a continuation of $a$, $a'$ that is $2|s_{n+1}|+2$ bits longer than $a$ and
$C_{n+1}(b'')\leq C_{n+1}(b')$ and $\tilde{b''}\subseteq E(a')$. If we find such $b''$, we set $F(x')=(a',b'')$.

If $s_{n+1}=\epsilon$, 
we define $\frac{1}{q_{n+1}}=2\frac{ \hat c_{n+1}(b)}{(2-C_n(b))}$ and set $Q(x')=q_{n+1}$.
We have $C_{n+1}(b)=C_n(b) + \frac{2-C_n(b)}{2} < 2$.
We then compute every measure $\mu_i, i\leq n+1$ for which $q_i>0$ on all continuations of $b'$ until we find
some $b''$ such that there is an $a'$ that extends $a$ by $2$ bits and
$C_{n+1}(b')\leq C_{n+1}(b)$ and $\tilde{b'}\subseteq E(a')$. If we find such $b'$, we set $F(x')=(a',b')$.

\begin{claim}\label{cl_H}
If all of the assumptions $H(s_1), H(x[1:2]),\dots, H(x),H(x')$ are true, for some string $s_{n+1}$ assumption $H(x')$ is true.
\end{claim}
\begin{proof}
There either exists some $b'\succ b$ such that $\mu_{n+1}$ is undefined and $C_n(b')<2$ (case \eqref{existsb}) or
not (case \eqref{allb}). Therefore, for some string $s_{n+1}$ assumption $H(x')$ is true.
\end{proof}

\begin{claim}\label{cl_v}
If all of the assumptions $H(s_1), H(x[1:2]),\dots, H(x),H(x')$ are true,
then for any $b'\succeq b$, such that $C_{n+1}(b')<2$ there is some bit $v$
such that $C_{n+1}(b'v)<2$.
\end{claim}
\begin{proof}
Let $I$ be the subset of indices $i\leq n+1$ such that $s_i\neq\epsilon$ (or, equivalently, $q_i>0$).
Order $I$ in increasing order, $I=\{i_1,\dots,i_z\}$, and denote
$(a_1,b_1)=F(x'[1:i_1]),\dots,(a_z,b_z)=F(x'[1:i_z])$.
We have that $b_1\prec\dots\prec b_z\preceq b\preceq b'$.
Since $C_{i_1-1}$ is always zero, condition $H(x[1:i_1])$ implies that $\mu_{i_1}$ is defined on all continuations of $b_1$,
and in particular on $b'0$ and $b'1$. Then $C_{i_1}(b'0)$ and $C_{i_1}(b'1)$ are both defined.
Since $C_{i_1}(b')\leq C_{n+1}(b') < 2$  and  $\frac{C_{i_1}(b'0) + C_{i_1}(b'1)}{2}=C_{i_1}(b')$,
we have that for at least one $v$, $C_{i_1}(b'v)<2$.
But then, by $H(x[1:i_2])$ , $\mu_{i_2}$ is defined on $b'v$, and by definition of measure it is also defined on 
$b'\bar{v}$. Again, for at least one $v$, $C_{i_2}(b'v)<2$.
And so on...
We have that $C_{n+1}(b'v)$ is defined for both values of $v$, and for at least one $C_{n+1}(b'v)<2$.
\end{proof}

\begin{claim}\label{cl_indstep}
If all of the assumptions $H(s_1), H(x[1:2]),\dots, H(x),H(x')$ are true, 
then both $F(x')$ and $Q(x')$ are defined,
and satisfy condition \eqref{CA}.
\end{claim}
\begin{proof}
In case $s_{n+1}$ is a nonempty string, $Q(x')=0$ and 
we have $C_{n+1}(bs_{n+1})<2$.
In the other case, $s_{n+1}$ is an empty string, $Q(x')>0$ and $C_{n+1}(b)<2$.
In both cases $Q(x')$ is defined.
Denote $b'=bs_{n+1}$.
In both cases $b\preceq b'$  and $C_{n+1}(b')<2$.
Let $\ell$ be the smallest number such that for every string $s$ of length $\ell$,
and any $a'\succ a$ with $|a'|-|a|=2|s_{n+1}|+2$, we have that $\tilde{s}$ is either a subset of $E(a')$ or disjoint from it.
By repeated use of claim \ref{cl_v} we can find $b''\succeq b'$ of length $\ell$
such that $C_{n+1}(b'')<2$.
We have that $\tilde{b''}$ is a subset of $E(a')$ for some
$a'\succ a$ with $|a'|-|a|=2|s_{n+1}|+2$. The pair $(a',b'')$ satisfies condition \eqref{CA}.
The value $F(x')$ is the first such pair that is found.
\end{proof}

For $F(\xi),Q(\xi)$, condition \eqref{CA} is satisfied.
By claims \ref{cl_H}, \ref{cl_indstep} and induction, 
there is an infinite sequence of strings $\chi = s_1,s_2,\dots$ such that
for all $n$, $H(\chi[1:n])$ is true and $F(\chi[1:n]), Q(\chi[1:n])$ are defined and satisfy \eqref{CA}.
We'll say that such $\chi$ is \textit{good}.
Denote $(a_n,b_n)=F(\chi[1:n])$.
Let the pair $(\alpha,\beta)$ denote the binary sequences such that for all $n$, 
$a_n\prec \alpha$ and $b_n\prec \beta$.
We'll say that the pair $(\alpha,\beta)$ is \textit{good}.
\begin{claim}
If $(\alpha,\beta)$  is a good pair, $\beta$ is PCR and $\beta\in E(\alpha)$.
\end{claim}
\begin{proof}
Clearly, $\beta\in E(\alpha)$.
The $i$-th partially computable measure $\mu_i$ is either undefined on $\beta$, or
is less than $C_i(\beta)/q_i$ which is less than $2/q_i$, and $\dot\mu_i$ does not win on $\beta$.
Therefore no partially computable martingale wins on $\beta$.
\end{proof}

All that remains to show is that $\alpha$ is non-MLR.
We define a function that assigns rational weights to sequences of strings. 
$w(s_1,\dots,s_n) = \prod_{i\leq n} 2^{-2|s_i|-1}$.
Let $G(a)$ denote the set of sequences of strings such that $s_1,\dots,s_n,s_{n+1}$ is in the set iff
$a'\prec a \preceq a''$ where $(a',b')=F(s_1,\dots,s_n)$ and $(a'',b'')=F(s_1,\dots,s_n,s_{n+1})$.
We define $\gamma(a)=\sum_{x\in G(a)} w(x)$.

\begin{claim}
$\gamma$ is a lower semicomputable semimeasure.
\end{claim}
\begin{proof}
Since $F$ is partially computable, $G(a)$ is computably enumerable and $\gamma$ is lower semicomputable.
We still need to show that for all $a$, $\gamma(a) \geq \gamma(a0) + \gamma(a1)$.
For any sequence of strings $x$, summing weights over all single-string continuations of $x$,
we have $$\sum_{s\in\{0,1\}^*} w(x,s)=w(x)\sum_{s\in\{0,1\}^*} 2^{-2|s|-1}=w(x)\sum_{\ell\geq 0} 2^\ell2^{-2\ell-1}=w(x)$$.

Let $x\in G(a)$ and $(a',b')=F(x[1:|x|-1])$, $(a'',b'')=F(x)$.
We split $G(a)$ into two sets, $G^\prec(a)$ and $G^=(a)$.
If $a'\prec a \prec a''$, then $x\in G^\prec(a)$.
On the other hand,  if $a'\prec a = a''$, then $x\in G^=(a)$.
For $x\in G^\prec(a)$, either $a'\prec a0 \preceq a''$ or $a'\prec a1 \preceq a''$ (but not both),
so either $x\in G(a0)$ or $x\in G(a1)$.
For $x\in G^=(a)$, $x$ is neither in $G(a0)$ nor in $G(a1)$.
Let $x\in G^=(a)$ and $s$ any string. 
Suppose  $F(x,s)$ is defined and denote $(a^*,b^*)=F(x,s)$. 
Then again, either 
$a''\prec a0 \preceq a^*$ or $a''\prec a1 \preceq a^*$ (but not both),
so either $x,s\in G(a0)$ or $x,s\in G(a1)$.
We have that $G(a0)$ and $G(a1)$ are disjoint,
and $\gamma(a0) + \gamma(a1) = \sum_{x\in G^\prec(a)} w(x) + \sum_{x\in G^=(a)}\sum_{s\in\{0,1\}^*, F(x,s)\text{is defined} } w(x,s)$.
This is less than the sum
$\sum_{x\in G^\prec(a)} w(x) + \sum_{x\in G^=(a)}\sum_{s\in\{0,1\}^*} w(x,s)$.
This  is equal to
$\sum_{x\in G^\prec(a)} w(x) + \sum_{x\in G^=(a)} w(x) = \sum_{x\in G(a)} w(x) = \gamma(a)$.
\end{proof}

\begin{claim}
If $(\alpha,\beta)$  is a good pair, $\alpha$ is non-MLR.
\end{claim}
\begin{proof}
We have that $|a_n|=\sum_{i\leq n} 2|s_i|+2$, and $\gamma(a) \geq w(\chi[1:n])= \prod_{i\leq n} 2^{-2|s_i|-1}$.
Then $\gamma(a_n)/\lambda(a_n)\geq2^n$, and $\dot\gamma$ wins on $\alpha$.
\end{proof}

We have shown that $\alpha$ is non-MLR, and $E(\alpha)$ contains $\beta$ that is PCR, a contradiction.
\end{proof}

\section{A computable encoding for KLR}
\begin{thm}\label{KLR=_EMLR}
There is a computable encoding $E$ and two total computable nonmonotonic betting strategies $A,B$ such that for every non-MLR sequence $\sigma$ and every $\rho$ in $E(\sigma)$, $A$ or $B$ win on $\rho$. 
\end{thm}
\begin{proof}[Proof by construction.]
We first define the encoding $E$.
\begin{definition}[Encoding]
For string $s$:\\ 
Function $v$ computes the integer value of $s$, $v(s):=2^0\cdot s[1] +\dots+ 2^{|s|-1}s[|s|]$\\
Function $p$ assigns a unique position to each $s$, $p(s):=2^{|s|}+v(s)$.\\
Function $e$ is $p$'s finite assignment for $s$, that is, $e(s):=\mathtt{fa}_p(s)$.\\
The encoder, $E$, is the (clopen) set of sequences consistent with the finite assignment for $s$, $E(\tilde s)=\tilde{e}(s)$.
\end{definition}
\begin{claim}\label{E_is_encoding}
Encoder preserves uniform Lebesgue measure and disjointedness relation 
\end{claim}
\begin{proof}
For any string $s$, $\lambda(E(\tilde s))=\lambda(\tilde e(s))=2^{-|e(s)|}=2^{-|s|}=\lambda(s)$.
For all strings $s_1,s_2$,
if $\tilde s_1$ is disjoint from $\tilde s_2$, then there is some 
string $s$ and bit $v$ such that $sv\preceq s_1$ and $s\bar v\preceq s_2$.
But then $e(s_1)$ assigns $v$  and $e(s_2)$ assigns $\bar v$ to position $p(s)$, and $e(s_1),e(s_2)$ are inconsistent.
Therefore  $E(\tilde s_1),E(\tilde s_2)$  are disjoint.
\end{proof}
\newcommand{\faa}{\mathtt{fa}_A}
\newcommand{\fab}{\mathtt{fa}_B}
\newcommand{\tfaa}{\tilde{\mathtt{fa}}_A}
\newcommand{\tfab}{\tilde{\mathtt{fa}}_B}
We will define the betting strategies $A,B$ in stages.
A stage starts with a pair of strings $a,b$, 
their finite assignments $\faa(a), \fab(b)$, 
and measures $\mu_A(a),\mu_B(b)$.
This stage will define the strategies $A,B$ on extensions of $a,b$.
As the stage computes $A,B$, for some of the pairs of extensions, the stage will stop, and the pair will exit the stage.
On such pairs, we start the next stage, and so on, defining $A,B$ entirely.
\newcommand{\faf}{\mathtt{fa}_F}
\newcommand{\far}{\mathtt{fa}_R}
\newcommand{\tfaf}{\tilde{\mathtt{fa}}_F}
\newcommand{\tfar}{\tilde{\mathtt{fa}}_R}
\begin{definition}[Stage]
A stage starts with a pair of strings
for which the finite assignments 
and measures of strategies $A,B$
have already been defined.
We require that the measures are non-zero and that strings are leafs of $A,B$ such that their finite assignments are equivalent, and even more, equivalent to a string.
Denote the input pair with $(a,b)$ and string $c\equiv\faa(a)\equiv\fab(b)$ .

Stage iteratively defines strategies $A,B$ on leafs extending $a,b$.
The capitals  on strings $a,b$ will determine how the bets are defined.
The strategy with the smaller capital will at first place bets on positions in the same order as the position function $p$, we'll call this the \textit{forward} strategy and label it with $F$.  
The strategy with the larger capital will at first skip the first position and place bets on other positions determined by $p$ and later return to the first position. 
We'll call this the \textit{rewind} strategy and label it with $R$.
So if $\dot\mu_A\leq\dot\mu_B$, then $F=A$, $R=B$, and $F,R$ will be defined on continuations of $f,r = a,b$.
In the other case, $F=B$, $R=A$, and $F,R$ will be defined on continuations of $f,r = b,a$.
We set the threshold $T$ for the capital that needs to be achieved
$T=\frac{3}{2}c_R(r)$.

We'll define the measure functions $\mu_F$, $\mu_R$ as sums of two measure functions. 
One, $\eta$, will assign measure evenly, this ensures that the strings always have non-zero measure assigned to them.
The other that we will define, $\nu$, will increase (or decrease) capital on some of the strings.
We set $\eta_F(\epsilon)=\nu_F(\epsilon)=\frac{1}{2}\mu_F(f)$ and
$\eta_R(\epsilon)=\nu_R(\epsilon)=\frac{1}{2}\mu_R(r)$.
We have
$\mu_F(f)=\eta_F(f)+\nu_F(f)$, $\mu_R(r)=\eta_R(r)+\nu_R(r)$.

Since the encoding's position function $p$, 
as the input strings become longer, chooses positions in increasing order, 
there exists a string $s$ such that $E(\tilde s)$ still contains $\tilde{c}$, and $p(s)>|c|$.
 We will computably enumerate a set of strings $M$, that are continuations of this $s$ and on which the universal lower semicomputable semimeasure $u$
 achieves capital larger than $H=\frac{2T}{ \lambda(s)\dot\nu_F(f)}$.
 More precisely, starting with empty $M$,
 we dovetail the computation of $u$ on all inputs, until finding $v\succ s$ such that
 $\dot u(v)\geq H$.
 If $\tilde v$ is already a subset of $\tilde M$, nothing new gets  added to $M$,
 otherwise, we will
 add one by one strings $s_1,\dots,s_n$ 
 such that $\tilde{s_1},\dots,\tilde{s_n}$ are disjoint and
 contain all of the sequences in $\tilde M\setminus\tilde{v}$.
 
 When we add a string $s'$ to $M$,
 we define finitely many new bets 
 in an \textit{iteration} of the stage .
Let $M'$ denote the strings that have been added in previous iterations.
The bets are defined only for the leafs of $F,R$ that don't extend strings that have exited the stage
after some previous iteration.
We denote these leafs $F^-,R^-$.
We define bets on strings that extend $F^-,R^-$ in three consecutive phases.
 \\
 \\ 
Phase F)\\
 For every $f^-\in F^-$, every leaf of $F$, $f'$, that extends $f^-$,
 if $\faf(f')$ is consistent, but doesn't refine $e(s')$,
 we define $\pi_F(f')$ to be the smallest position in $\pos(e(s'))$ on which $\faf(f')$ doesn't have assigned value.
 The mass is split evenly, $\nu_F(f'0)=\nu_F(f'1)=\frac{1}{2}\nu_F(f')$.
 We keep doing this until for every leaf $f'$ of $F$ that extends some $f^-\in F^-$,
 $\faf(f')$ is inconsistent with $e(s')$, or refines it.
\\
\\
Phase R)\\
For every $r^-\in R^-$, every leaf of $R$, $r'$, that extends $r^-$,
 if $\far(r')$ is consistent, but doesn't refine $e(s')$,
 we define $\pi_R(r')$ to be the smallest position larger than $p(s)$ in $\pos(e(s'))$ on which $\faf(r')$ doesn't have assigned value, if there is such a position.
 In this case, the mass is split evenly, $\nu_R(r'0)=\nu_R(b'1)=\frac{1}{2}\mu_R(r')$.
 In the other case, $p(s)$ is the only position in $\pos(e(s'))$ on which $\faf(r')$ doesn't have assigned value.
 We define $\pi_R(r')=p(s)$. Let $v=s'[|s|+1]$, and set $\nu_R(r'v)=\nu_R(r')$, $\nu_R(r'\bar{v})=0$. 
 We repeat defining bets like this until for every leaf $r'$ of $R$ that extends some $r^-\in R^-$ we have that 
 $\far(r')$ is inconsistent from $e(s')$, or refines it.
\\

Let $F',R'$ denote the leafs extending $F^-,R^-$ after phases F), R), whose finite assignments are inconsistent
with $e(s')$, and $F^+,R^+$ the ones that refine $e(s')$. Let $Z$ be the strings in $R^-$ that have measure $\nu_R$ zero, and denote the set of sequences consistent 
 with their finite assignments with $\tilde Z = \bigsqcup_{z\in Z} \tfar(z)$.
 \\
 \\
Phase C)\\
 Let $g$ be the maximal position in finite assignments of leafs of both strategies.
Let $r^+$ be a leaf in $R^+$,
Let $G$ be a subset of $[1,g]$ on which  $\far(r^+)$ does not assign a value.
For each string $x$ of length $|G|$, we define $\pi_R$ so that $\far(r^+x)$ assigns values to all positions in $[1,g]$. The mass $\nu_R$ is distributed evenly,
$\dot\nu_R(r^+x)=\dot\nu_R(r^+)$.
 Let $f^+$ be a leaf in $F^+$.
 Like before, let $G$ be a subset of $[1,g]$ on which  $\faf(f^+)$ does not assign a value,
 For each string $x$ of length $|G|$, we define $\pi_F$ so that $\faf(f^+x)$ assigns values to all positions in $[1,g]$.
 This is the same as what we did in this phase for $R$ when defining $\pi_R$.
 However, the measure $\nu_F$ is not always distributed evenly.
 
 In case $\tfaf(f')$ intersects $\tilde Z$ we assign measure the following way.
 If $\tfaf(f'x)\subseteq \tilde Z$ then $\nu_F(f'x)=\frac{2^{-|x|}}{\lambda(\tilde Z|\tfaf(f'))}\nu_F(f')$, and if not
 then $\nu_F(f'x)=0$.
 In case $\tfaf(f')$ is disjoint from $\tilde Z$, the measure $\nu_F$ is distributed evenly among continuations of $f'$.
\\
 
 Let $F^\circ,R^\circ$ be the leafs extending strings in $F^+,R^+$, after phase C).
 The strings in $F^\circ,R^\circ$ exit the stage, $s'$ is added to $M'$, and the stage then continues with the 
 adding of strings to $M$.
\end{definition}

In the remainder of the proof, we use the notation from the above definition, and
in addition to the notation used for an arbitrary iteration,
for $i$-th iteration we denote the input string with $s_i$,
$M_{i-1}=\{s_1,\dots,s_{i-1}\}$  is the set strings that were enumerated into $M$ in the first $i-1$ iterations,
$F^-_{i-1},R^-_{i-1}$  the leafs in $F,R$, before any of the bets for $s_i$ were defined, that do not extend some string that exited stage after the first $i-1$ iterations,
$F'_i,R'_i,F^+_i,R^+_i$ the leafs extending $F^-_{i-1},R^-_{i-1}$  after phases F) R) are finished, and
with $F^\circ_i,R^\circ_i$ the leafs that exit the stage.

From the definition of stage we have $M_0=\{\}, F^-_0=\{f\}, R^-_0=\{r\}$, 
and $\tilde F^+_{i} \sqcup \tilde F'_{i}=\tilde F^-_{i-1}$, $\tilde R^+_{i} \sqcup \tilde R'_{i}=\tilde R^-_{i-1}$,
$\tilde F^\circ_i=\tilde F^+_i$, $\tilde R^\circ_i=\tilde R^+_i$,
$F^-_i=F'_i$, $R^-_i=R'_i$, $M_i=M_{i-1}\cup\{s_i\}$.

\begin{claim}\label{cl_F}
Let $f^-\in F^-$.
\begin{enumerate}[label=(\arabic*), ref=\theclaim(\arabic*)]
\item If $\faf(f^-)$ inconsistent with $e(s')$ then $F'|f^-=\{f^-\}$.\label{cl_F:f-inc_e}

\item If $\faf(f^-)$ consistent with $e(s')$, then for 
every $f'\in F'|f^-$ there is some
string $x$, and bit $v$ such that $s\preceq xv\preceq s'$  and \\ 
$\faf(f')=\faf(f),e(x\bar v)[|s|+1:|x\bar v|]$ and 
$\dot\nu_F(f')=\dot\nu_F(f)$.\label{cl_F:f-con_e_f'}

\item If $\faf(f^-)$ consistent with $e(s')$, then there is exactly one 
$f^+\in F^+|f^-$, and $\faf(f^+)=\faf(f),e(s')[|s|+1,|s'|]$ and 
$\dot\nu_F(f^+)=\dot\nu_F(f)$. \label{cl_F:f-con_e_f+}
 

\item If $\tfaf(f^+)$ intersects $\tilde Z$ then for all
$f^\circ\in F^\circ|f^+$,
if, for some $r^-\in Z$, $\faf(f^\circ)$ refines $\far(r^-)$
then $\dot\nu_F(f^\circ)=\dot\nu_F(f^-)/\lambda(\tilde Z| \tfaf(f^+))$
else $\nu_F(f^\circ)=0$. \label{cl_F:f-con_e_Z_fo}
\end{enumerate}
\end{claim}
\begin{proof}
By definition of phase F), if $\faf(f^-)$ is inconsistent with $e(s')$, then no new bets are defined for $f^-$, and
$f^-$ does not exit the stage, therefore \ref{cl_F:f-inc_e}.

For $i=0$, $F^-_i=\{f\}$, and we have that for every $f^-\in F^-_i$, $f^-$ fulfills the condition that
there is some string $y\succeq s$  such that $\faf(f^-)$ is 
equal to $\faf(f),e(y)[|s|+1:|y|]$ and $\dot\nu_F(f^-)=\dot\nu_F(f)$, to see this set $y=s$.
Assume that the same is true for some $i>0$.
Let $f^-\in F^-_i$.
If $\faf(f^-)$ is inconsistent with $e(s_{i+1})$, then by \ref{cl_F:f-inc_e},
$F'_i|f^-=\{f^-\}$, and the condition is true for all $f'\in F'_i|f^-$.
If $\faf(f^-)$ is consistent with $e(s_{i+1})$, then for some $y\preceq s_{i+1}$
$\faf(f^-)=\faf(f),e(y)[|s|+1:|y|]$ and $\dot\nu_F(f^-)=\dot\nu_F(f)$.
Let $\ell=|s_{i+1}|-|y|$, and let $v_1, \dots, v_\ell$ be the sequence of bits that extend $y$ to $s_{i+1}$,
that is, $s_{i+1}=y v_1 \dots v_{\ell}$
By definition of phase F),
$F'_{i+1}|f^-=\{f^-\bar v_1,f^-v_1\bar v_2,\dots, f^- v_1\dots v_{\ell-1}\bar v_{\ell} \}$,
and for every $f'\in F'|f^-$, there is some $n\leq \ell$, $f'=f^- v_1\dots v_{n-1}\bar v_{n} $ 
and $\faf(f' )$ is 
equivalent to $\faf(f),e(yv_1\dots v_{n-1}\bar v_{n})[|s|+1:|y|+n]$ and $\dot\nu_F(f')=\dot\nu_F(f)$.
Since $F'_{i+1}=\bigsqcup_{f^-\in F^-_i} F'_{i+1}|f^-$, and $F^-_{i+1}=F'_{i+1}$,
we have that for every $f^-\in F^-_{i+1}$, $f^-$ fulfills the condition, and \ref{cl_F:f-con_e_f'} follows by induction.

By \ref{cl_F:f-con_e_f'} for $f^-\in F^-$
there is some string $y\succeq s$  such that $\faf(f^-)$ is 
equal to $\faf(f),e(y)[|s|+1:|y|]$ and $\dot\nu_F(f^-)=\dot\nu_F(f)$.
If $f^-$ is consistent with $e(s')$ then $y\preceq s'$. 
Let $w$ denote the extension of $y$ to $s'$, that is, $yw=s'$.
By definition of phase F), $f^+=f^-w$, and
$\faf(f^+)=\faf(f),e(yw)[|s|+1:|yw|]$ and 
$\dot\nu_F(f^+)=\dot\nu_F(f)$, proving \ref{cl_F:f-con_e_f+}.

From definition of phase C), 
$f^\circ\in F^\circ$ assigns to all positions up to $g$,
the largest position assigned by a leaf of either strategy, before phase C) started.
Then for any $r^-\in R^-$, $\faf(f^\circ)$ either refines $\far(r^-)$, or is disjoint from it.
Therefore if $\tfaf(F^\circ)\subseteq \tilde Z$, there is some $r^-\in Z$, such that
$\faf(f^\circ)$ refines $\far(r^-)$.
The value for $\nu_F(f^\circ)$ in
\ref{cl_F:f-con_e_Z_fo} then follows from definition of phase C).

\end{proof}

\begin{claim}\label{cl_R}
Let $r^-\in R^-$.

\begin{enumerate}[label=(\arabic*), ref=\theclaim(\arabic*)]
\item If $\far(r^-)$ inconsistent with $e(s')$ then $R'|r^-=\{r^-\}$.\label{cl_R_inc_e}

\item If $\far(r^-)$ consistent with $e(s')$ and $p(s)\in\pos_R(r^-)$
then for every $r'\in R'|r^-$, $\dot\nu_R(r')=\dot\nu_R(r^-)$.
For $r^+\in R^+|r^-$, $\far(r^+)$ refines $e(s')$ and $\dot\nu_R(r^+)=\dot\nu_R(r^-)$.
For every $r^\circ\in R^\circ|r^+$, $\dot\nu_R(r^\circ)=\dot\nu_R(r^-)$. \label{cl_R_con_e_p}

\item If $\far(r^-)$ consistent with $e(s')$ and $p(s)\not\in\pos_R(r^-)$,
then for every $r^+\in R^+|r^-$, $\far(r^+)$ refines $e(s')$
and $\dot\nu_R(r^+)=2\dot\nu_R(r^-)$, 
and for all $r^\circ\in R^\circ|r^+$ $\dot\nu_R(r^\circ)=\dot\nu_R(r^+)$. \label{cl_R_con_e_nop_r+}

\item If $\far(r^-)$ consistent with $e(s')$ and $p(s)\not\in\pos_R(r^-)$,
there is exactly one string $z$ in $R'|r^-$
such that $\pos(e(s'))\subseteq \pos_R(z)$. $\far(z)$ assigns negation of $s'[|s|+1]$
to position $p(s)$ and $\nu_R(z)=0$.
For all other $r'\in R'|r^-$,
$\far(r')$ does not assign a value to $p(s)$, and $\dot\nu_R(r')=\dot\nu_R(r^-)$.  \label{cl_R_con_e_nop_r'}
\end{enumerate}
\end{claim}
\begin{proof}
By definition of phase R), if $\far(r^-)$ is inconsistent with $e(s')$, then no new bets are defined for $r^-$, and
$r^-$ does not exit the stage, therefore \ref{cl_R_inc_e}.

By definition of phase R), $\dot\nu_R$ changes only in the bet on position $p(s)$, 
if $\far(r^-)$ already assigns to $p(s)$, $\dot\nu_R$  remains the same for all strings that extend it.
In phase R) new bets are defined for strings that have finite assignment consistent with $e(s')$,
until all of the positions in $\pos(e(s'))$ are assigned to, therefore $\pos(e(s')) \subseteq \pos_R(r^+)$
and
$\far(r^+)$ is consistent with $e(s')$ implies $\far(r^+)$ refines $e(s')$, proving \ref{cl_R_con_e_p}.

On the other hand, if $\far(r^-)$ does not assign a value to $p(s)$, then there is
a single string $x$ that extends $r^-$ such that $\far(x)$ is consistent with $e(s')$ 
and assigns to all positions in $\pos(e(s'))$ except $p(s)$.
Let $v$ be the value that $e(s')$ assigns to position $p(s)$, that is, $v=s'[|s|+1]$.
A new bet is defined for $x$ on position $p(s)$, 
then $\{xv\}=R^+|r^-$, $\far(xv)$ refines $e(s')$ and $\dot\nu_R(xv)=2\dot\nu_R(r^-)$,
and $\far(x\bar v)$ inconsistent with $e(s')$ and $\nu_R(x\bar v)=0$,
proving \ref{cl_R_con_e_nop_r'},\ref{cl_R_con_e_nop_r+}.
\end{proof}

\begin{claim}\label{cl_R-}
For every $r^-\in R^-$ either $\far(r^-)$
doesn't assign to $p(s)$ and has $\dot\nu_R(r^-)=\dot\nu_R(r)$
or $\far(r^-)$ assigns to $p(s)$ and has $\nu_R(r^-)=0$.
\end{claim}
\begin{proof}
For $i=0$, we have that $R^-_i=\{r\}$ and for every $r^-\in R^-_i$ either $\far(r^-)$
doesn't assign to $p(s)$ and has $\dot\nu_R(r^-)=\dot\nu_R(r)$
or $\far(r^-)$ assigns to $p(s)$ and has $\nu_R(r^-)=0$.
Assume this is true for some $i>0$.
If for $r^-\in R^-_i$ $\far(r^-)$ assigns to $p(s)$ and has $\nu_R(r^-)=0$,
then for any string $r'$ that extends $r^-$, and in particular
all $r'\in R'|r^-$, $\far(r')$ assigns to $p(s)$ and  $\nu_R(r^-)=0$.
On the other hand, if for $r^-\in R^-_i$ $\far(r^-)$ 
doesn't assign to $p(s)$ and has $\dot\nu_R(r^-)=\dot\nu_R(r)$,
in case $\far(r^-)$ is inconsistent with $e(s')$, by claim \ref{cl_R_inc_e}, the same is true 
for all $R'|r^-$.
In case $\far(r^-)$ is consistent with $e(s')$, by claim \ref{cl_R_con_e_nop_r'},
for every $r'\in R'|r^-$
either $\far(r')$
doesn't assign to $p(s)$ and has $\dot\nu_R(r')=\dot\nu_R(r)$
or $\far(r')$ assigns to $p(s)$ and has $\nu_R(r')=0$.
Since $R'_{i+1}=\bigsqcup_{r^-\in R^-_i} R'|r^-$, and $R^-_{i+1}=R'_{i+1}$.
we have that 
the statement is true for $i+1$, and by induction for arbitrary iteration.
\end{proof}

\begin{lemma}\label{lem_Z|f+}
For $f^+\in F^+$, $\lambda(\tilde Z|\tfaf(f^+))\leq \frac{\dot\nu_F(f)}{T}$
\end{lemma}
\begin{proof}
Note that by claim \ref{cl_R}, the only case when a string is assigned $\nu_R$ zero is in case
\ref{cl_R_con_e_nop_r'}.
Therefore, if $r'\in Z$ then for some previous iteration $i$
there is an $r^-_i,r'_i\in R^-_i,R'_i$, $r^-_i\preceq r'_i\preceq r'$, $\nu_R(r^-_i)>\nu_R(r_i)=0$,
and
$\far(r_i)$ assigns $\bar v$ to  $p(s)$, where $v$ is the $|s|+1$-st bit of $s_i$ 
and
we have $\pos(e(s_i))\subseteq \pos_R(r_i)\subseteq \pos_R(r')$.
Let $Z_i$ denote the subset of $Z$ such that for $r'\in Z_i$, there is
$r^-_i,r^+_i\in R^-_i,R^+_i$, $r^-_i\preceq r^+_i\preceq r'$, $\nu_R(r^-_i)>\nu_R(r_i)=0$.
Suppose we are in the $n$-th iteration. We have $Z=\bigsqcup_{i<n} Z_i$.
Then
$\lambda(\tilde Z\cap\tfaf(f^+))=\sum_{i<n}\lambda(\tilde Z_i\cap\tfaf(f^+))$.

By claim \ref{cl_F:f-con_e_f+} $\faf(f^+)$ is equivalent to $\faf(f),e(s')[|s|+1:|s'|]$,  and for all $i<n$ it is inconsistent with $e(s_i)$,  and $\pos_F(f^+)=\pos(e(s'))\cup [1,|c|]$.
Let $r'\in Z_i$  such that $\far(r'),\faf(f^+)$ consistent, and let $v$ be the $|s|+1$-st bit of $s_i$.
We have that $|s|+1$-st bit of $s'$ is $\bar v$, and since encoding's position function $p$ assigns a unique position to each string,
$|\pos(e(s_i))\setminus \pos(e(s_n))|=|s_i|-|s|-1$.
Denote with  $h_i$ the last $|s_i|-|s|-1$ assignments of $e(s_i)$, that is,
$h_i=e(s_i)[|s|+2:|s_i|]$. We have that $\pos(h_i)$ is disjoint from $\pos_F(f^+)$ and every $r'\in Z_i$ refines
$h_i$. Therefore
$\lambda(\tilde Z\cap\tfaf(f^+))\leq\sum_{i<n}\lambda(\tilde h_i)\lambda(\tfaf(f^+))=2\frac{\lambda(\tilde M_{n-1})}{\lambda(s)} \lambda(\tfaf(f^+))$.
The result follows from $\lambda(\tilde M_{n-1})\leq \lambda(\tilde M)\leq \frac{1}{H}=\frac{\lambda(s)\dot\nu_F(f)}{2T}$.
\end{proof}
%
%
For a set of mutually inconsistent strings $X$ and strategy $Y$ 
we abbreviate $\bigsqcup_{x\in X}\tilde{\mathtt{fa}}_Y(x)$ 
to $\tilde{\mathtt{fa}}_Y(X)$.

\begin{lemma}\label{lem_FR_exit_seqs}
$\tfaf(F^\circ) =\tfar(R^\circ)= E(\tilde s')\cap \tilde c$, and
for every $f^\circ \in F^\circ$ there is a $r^\circ \in R^\circ$ and
string $c'\succ c$
such that $\faf(f^\circ)\equiv\far(r^\circ)\equiv c'$,
and $\max(\dot\mu_F(f^\circ),\dot\mu_R(r^\circ))\geq T$,
$\min(\dot\mu_F(f^\circ),\dot\mu_R(r^\circ))\geq 0$.
\end{lemma}
\begin{proof}
We have $\tilde F^-=\tilde F^+ \sqcup \tilde F'$,
every $f'\in F'$ has $\faf(f')$ inconsistent with $e(s')$,
$\tilde F^+ = \tilde F^\circ$ implies $\tfaf (F^+)=\tfaf (F^\circ)$
therefore
$\tfaf (F^\circ) = \tfaf (F^-)\cap E(\tilde s')$. 
Similarly, 
$\tfar (R^\circ) = \tfar (R^-)\cap E(\tilde s')$.
For first iteration,
$\tfaf (F_1^\circ) = \tfaf (F_0^-)\cap E(s_1)=\tilde c \cap E(s_1)=\tfar (R_0^-)\cap E(\tilde s_1) = \tfar (R_1^\circ)$.
Assume that for all iterations $i\leq n$ we have  $\tfaf(F_i^\circ) =\tfar(R_i^\circ)= E(\tilde s_i)\cap \tilde c$.
We have $\tfaf(F^-_n)=\tfaf(R^-_n)=\tilde c \setminus \sqcup_{i\leq n}E(\tilde s_i)$.
Since $ \sqcup_{i\leq n} E(\tilde s_i)$ is disjoint from $E(\tilde s_{n+1})$, we have
$\tfaf(F_{n+1}^\circ) =\tfar(R_{n+1}^\circ)= E(\tilde s_{n+1})\cap \tilde c$.
By induction, we have that for arbitrary iteration, $\tilde F^\circ = \tilde R^\circ = \tilde c \cap E(s')$.

Since for any $f^\circ\in F^\circ,r^\circ\in R^\circ$, $\faf(f^\circ),\far(r^\circ)$ assign to positions in $[1,g]$, $\faf(f^\circ),\far(r^\circ)$ are either inconsistent or equivalent.
Then it cannot be that for some $f^\circ\in F^\circ$ there is no $r^\circ\in R^\circ$ with an equivalent finite assignment,
because then $\tfaf(F^\circ) \neq \tfar (R^\circ)$.

It remains to show that for each pair $f^\circ\in F^\circ,r^\circ\in R^\circ$
where $\faf(f^\circ)\equiv\far(r^\circ)$ we have
 $\max(\dot\nu_F(f^\circ),\dot\nu_R(r^\circ))\geq T$.
 
By claim \ref{cl_F:f-con_e_f'}, for all $f^-\in F^-$, $\dot\nu_F(f^-)=\dot\nu_F(f)$.
Let $f^\circ\in F^\circ, r^\circ\in R^\circ$ and $\faf(f^\circ)\equiv \far(r^\circ)$.
Then there is $f^-\in F^-, f^+\in F^+|f^-,r^-\in R^-$ 
such that $f^\circ\in F^\circ|f^+, r^\circ\in R^\circ|r^-$
and $\faf(f^-),\far(r^-)$ both consistent with $e(s')$.
If $\far(r^-)$ does not assign to $p(s)$, then by claim \ref{cl_R-},
$\dot\nu_R(r^-)=\dot\nu_R(r)$, and by claim \ref{cl_R_con_e_nop_r+},
$\dot\nu_R(r^\circ)=2\dot\nu_R(r)=T$.
On the other hand, if $\far(r^-)$ does  assign to $p(s)$, then by claim \ref{cl_R-},
$\nu_R(r^-)=0$. Since for arbitrary $i$, $F'_i=F^-_i$, by claim \ref{cl_F:f-con_e_f'}, $\dot\nu_F(f^-)=\dot\nu_F(f)$ and by claim \ref{cl_F:f-con_e_Z_fo},
$\dot\nu_F(f^\circ)=\dot\nu_F(f)/\lambda(\tilde Z |\faf(f^+))$.
This is by lemma \ref{lem_Z|f+} larger than $T$.
Therefore $\max(\dot\nu_F(f^\circ),\dot\nu_R(r^\circ))\geq T$.
We have that $\dot\eta_F(f^\circ)=\dot\eta_F(f)$, $\dot\eta_R(r^\circ)=\dot\eta_R(r)$.
Since $\mu$ is a sum of $\eta$ and $\nu$, and the stage requires that  $\mu_F(f),\mu_R(r)$ are larger than zero,
the result follows.
\end{proof}

If $F=A,R=B$
let $C_i$ denote the set of all pairs of strings $(a', b')\in F^\circ_i \times R^\circ_i$
such that $\faa(a'),\fab(b')$ are equivalent to a same string.
If $F=B,R=A$, then $C_i$ are pairs $(a', b')\in R^\circ_i \times F^\circ_i$ that are equivalent to a same string.
For a pair in $C_i$ we say that the pair exits the stage in $i$-th iteration.
Let $C$ be the set of pairs that exit the stage in some iteration, $C=\bigcup_{i}C_i$, 
and let $\tilde C$ be the computably enumerable open set of 
sequences consistent with finite assignments of strings that exit the stage,
$\tilde C = \bigsqcup_{(a', b')\in C}\tfaa(a')$.

\begin{proposition}\label{prop_stage}
Let $(a,b)$ be a pair of strings on which a stage is run.
Let $c$ be the string equivalent to $\faa(a),\fab(b)$.
$\tilde C$ contains all sequences that extend $c$ and are contained in an encoding of a non-MLR sequence.
For each $(a',b')\in C$, $\max(\dot\mu_A(a'),\dot\mu_B(b'))\geq\frac{3}{2}\max(\dot\mu_A(a),\dot\mu_B(b))$,
$\min(\dot\mu_A(a'),\dot\mu_B(b'))\geq 0$,
and there is some string $c'\succeq c$ such that $\faa(a')\equiv\fab(b')\equiv c'$.
\end{proposition}
\begin{proof}
Let  $\sigma$ be a non-MLR sequence that extends $s$.
Let $\rho$ be a sequence that extends $c$ and is contained in $E(\sigma)$.
Since $\sigma$ is non-MLR, it must have a prefix $s'$ that extends $s$,
and $\dot u(s')\geq H$. Then $s'\in M$.
We have
$\tilde C =\bigcup_{i} \bigsqcup_{(a', b')\in C_i}\tfaa(a') $.
By lemma \ref{lem_FR_exit_seqs}, 
$\tilde C =\bigcup_{i}\tilde c\cap E(s_i)$.
This is equal to $\tilde c\cap \bigcup_{s'\in M}E(\tilde s')$, and $\tilde C$ contains $\rho$.
The rest of the result follows from that same lemma.
\end{proof}

We recursively define entire $A,B$.
\begin{definition}
We set $\mu_A(\epsilon)=\mu_B(\epsilon)=1$
and $\tfaa(\epsilon)=\tfab(\epsilon)=\tilde{\epsilon}$
and run a stage on $(\epsilon,\epsilon)$.
For a pair of strings that exits any running stage,
we run another stage on the pair.
\end{definition}

Suppose that some $\rho\in E(\sigma)$ is consistent with finite assignments of only finitely many pairs that have exited some stages.
For the last  such pair, the stage is run on the pair, and $\rho$ is consistent with the pair's finite assignments,
but not with the finite assignments of pairs that exit the stage. By proposition \ref{prop_stage},
$\rho$ is not contained in the encoding of a non-MLR sequence.

Therefore if $\sigma$ is non-MLR, for every $\rho\in E(\sigma)$ 
there is an infinite sequence of pairs 
$(\epsilon,\epsilon),(a_1,b_1),\dots$
such that for all $n$
$(a_n,b_n)$ exits the stage run on $(a_{n-1},b_{n-1})$ and 
$\rho$ is consistent with their finite assignments.
By proposition \ref{prop_stage},
$\max(\dot\mu_A(a_n),\dot\mu_B(b_n))\geq\frac{3}{2}^n$.
Then either $\sup_n(\dot\mu_A(a_n))=\infty$, or $\sup_n(\dot\mu_B(b_n))=\infty$ (or both),
and one or both strategies win on $\rho$.
Since the set of non-Martin-L\"of sequences is everywhere dense, the strategies are total.
\end{proof}

\section{Concluding remarks}
We have constructed a total computable encoding $E$
that allows us to construct a test for Martin-L\"of randomness
from a Kolmogorov-Loveland randomness test.
A sequence $\alpha$ is Martin-L\"of random iff  a sequence in $E(\alpha)$ is Kolmogorov-Loveland random.
This result goes against the Schnorr's objection on lower-semicomputability of Martin-L\"of randomness.

We have also shown that no such encoding exists for partial computable randomness and
it is an
\begin{openq}
Which other known notions of randomness implied by Martin-L\"of randomness have this property?
\end{openq}

The construction of stages is similar to the proof of \cite[Theorem 5.9]{vanLambalgen} insofar that it involves constructing
two probability distributions on binary sequences that are largely independent, except on the basic sets enumerated into a level of a universal Martin-L\"of test. For these basic sets one of the distributions has density (capital) above a threshold, and then the distribution pair is further constructed, recursively, on these basic sets. 
Note that a similar construction cannot be done for Kolmogorov-Loveland randomness (without the encoding), \cite{TP}.
The constructions in \cite[Theorem~9.1]{Muchnik},  \cite[Theorem~12]{Merkle}, also construct two largely independent probability distributions, however the effective nullset on which they win is not universal.
There, like in our construction, the positions are being split into two sets, where one strategy first bets on positions in the first set, and the other one first bets on positions in the second set. However, the splitting is fixed, and not chosen adaptively.


\begin{thebibliography}{00}	
\bibitem{ShenBook} A. Shen, V. A. Uspensky, and N. Vereshchagin: Kolmogorov Complexity and Algorithmic Randomness,
	ISBN-10: 1-4704-3182-3
	
\bibitem{MartinLof}P. Martin-L\"of. The definition of random sequences.
	Inform. Contr., 9:602\textendash 619, 1966.

\bibitem{Schnorr}C.P. Schnorr. Zuf\"alligkeit und Wahrscheinlichkeit; Eine algorithmische Begr\"undung der Wahrscheinlichkeitstheorie, volume 218 of Lect. Notes Math. Springer-Verlag, Berlin, 1971.

\bibitem{Klaus}Klaus Ambos-Spies. Algorithmic randomness revisited, in:B.McGuinness(Ed.), Language,Logic and Formalization of Knowledge. Coimbra Lecture and Proceedings of a Symposium held in Siena in September 1997, Bibliotheca, 1998, pp.33–52.

\bibitem{Muchnik}Andrei A. Muchnik, Alexei L. Semenov, Vladimir A.
	Uspensky: Mathematical Metaphysics of Randomness. Theor. Comput. Sci.
	207(2): 263-317 (1998)
	
\bibitem{Merkle}Wolfgang Merkle, Joseph S. Miller, Andr\'e Nies, Jan
	Reimann, Frank Stephan: Kolmogorov-Loveland randomness and stochasticity.
	Ann. Pure Appl. Logic 138(1-3): 183-210 (2006)
	
\bibitem{Downey}Rodney G. Downey, Denis R. Hirschfeldt: Algorithmic
	Randomness and Complexity. Theory and Applications of Computability,
	Springer 2010, ISBN 978-0-387-95567-4, pp. 1-766
	
\bibitem{Nies} André Nies. Computability and randomness, volume 51 of Oxford Logic Guides. Oxford University Press, Oxford, 2009.

\bibitem{vanLambalgen} Van Lambalgen, M. (1990). The axiomatization of randomness. Journal of Symbolic Logic, 55(3), 1143-1167.

\bibitem{TP}Tomislav Petrović,
Kolmogorov-Loveland betting strategies lose the Betting game on open sets,
Theoretical Computer Science,
Volume 1037,
2025,
115177,
ISSN 0304-3975
\end{thebibliography}
\end{document}